\documentclass[sigconf]{acmart}
\AtBeginDocument{%
  }

\usepackage{multirow}
\usepackage{amsthm}

\theoremstyle{remark}

\copyrightyear{2025}
\acmYear{2025}
\setcopyright{acmlicensed}\acmConference[MM '25]{Proceedings of the 33rd ACM International Conference on Multimedia}{October 27--31, 2025}{Dublin, Ireland}
\acmBooktitle{Proceedings of the 33rd ACM International Conference on Multimedia (MM '25), October 27--31, 2025, Dublin, Ireland}
\acmDOI{10.1145/3746027.3755540}
\acmISBN{979-8-4007-2035-2/2025/10}

\begin{document}

\title{FITMM: Adaptive Frequency-Aware Multimodal Recommendation via Information-Theoretic Representation Learning}


\author{Wei Yang}
\authornote{These authors contributed equally to this work.}
\affiliation{%
  \institution{Kuaishou Technology}
  \city{Beijing}
  \country{China}
}
\email{yangwei08@kuaishou.com}
\author{Rui Zhong}
\authornotemark[1]
\affiliation{
  \institution{Kuaishou Technology}
  \city{Beijing}
  \country{China}
}
\email{zhongrui@kuaishou.com}

\author{Yiqun Chen}
\affiliation{
  \institution{Renmin University of China}
  \city{Beijing}
  \country{China}
}
\email{chenyiqun990321@ruc.edu.cn}

\author{Shixuan Li}
\affiliation{
  \institution{University of Southern California}
  \city{Los Angeles, CA}
  \country{USA}
}
\email{sli97750@usc.edu}

\author{Heng Ping}
\affiliation{
  \institution{University of Southern California}
  \city{Los Angeles, CA}
  \country{USA}
}
\email{hping@usc.edu}

\author{Chi Lu}
\affiliation{
  \institution{Kuaishou Technology}
  \city{Beijing}
  \country{China}
}
\email{luchi@kuaishou.com}

\author{Peng Jiang}
\affiliation{
  \institution{Kuaishou Technology}
  \city{Beijing}
  \country{China}
}
\email{jiangpeng@kuaishou.com}

\renewcommand{\shortauthors}{Wei Yang et al.}


\begin{abstract}
Multimodal recommendation aims to enhance user preference modeling by leveraging rich item content such as images and text. Yet dominant systems fuse modalities in the spatial domain, obscuring the frequency structure of signals and amplifying misalignment and redundancy. We adopt a spectral information-theoretic view and show that, under an orthogonal transform that approximately block-diagonalizes bandwise covariances, the Gaussian Information Bottleneck objective decouples across frequency bands, providing a principled basis for separate-then-fuse paradigm. Building on this foundation, we propose \textbf{FITMM}, a Frequency-aware Information-Theoretic framework for multimodal recommendation. FITMM constructs graph-enhanced item representations, performs modality-wise spectral decomposition to obtain orthogonal bands, and forms lightweight within-band multimodal components. A residual, task-adaptive gate aggregates bands into the final representation. To control redundancy and improve generalization, we regularize training with a frequency-domain IB term that allocates capacity across bands (Wiener-like shrinkage with shut-off of weak bands). We further introduce a cross-modal spectral consistency loss that aligns modalities within each band. The model is jointly optimized with the standard recommendation loss. Extensive experiments on three real-world datasets demonstrate that FITMM consistently and significantly outperforms advanced baselines. The source code is available at: \url{https://github.com/llm-ml/FITMM}.
\end{abstract}




\maketitle

\section{INTRODUCTION}
With the surge of multimodal content in e-commerce, short video platforms, and social media, users are increasingly exposed to items described through diverse modalities such as product images, textual descriptions, videos, and reviews \cite{han2022modality,cen2020controllable,he2021click,zhou2023comprehensive,yang2024enhancing}. This has led to the emergence of Multi-Modal Recommender Systems (MMRS), which aim to leverage heterogeneous sources to better capture user preferences and improve recommendation performance~\cite{wei2019mmgcn,deldjoo2020recommender,liu2024rec, malitesta2024formalizing,yang2025structured}. Effectively modeling and integrating such multimodal information has become a central challenge in modern recommender system research \cite{wang2023cl4ctr,liu2024multimodal,li2025cross,liu2025joint,xu2025survey,zhao2025hierarchical,gu2025r4ec}.

Recent advances have centered on Graph Neural Networks (GNNs) to capture the high-order collaborative signals in user-item graphs \cite{liu2022multimodal,guo2024lgmrec,zhou2023enhancing,li2024multimodal,kim2024self}. The dominant paradigm, however, involves a critical flaw: after encoding, modalities are fused in the spatial domain using simple gating or attention mechanisms. This approach struggles with the inherent statistical discrepancies and semantic misalignment across heterogeneous modalities (e.g., image vs. text) \cite{yang2023modal, shang2024improving}. By forcing a direct projection into a shared space, these methods often lead to one modality's signal overpowering another or critical information being diluted during fusion \cite{li2024align,yi2024unified,bai2024multimodality}.

More fundamentally, this paradigm overlooks the intrinsic frequency domain structure of multimodal signals. Information is not monolithic; it exists on a spectrum where different frequencies carry distinct semantic weight. Low-frequency components typically encode stable, cross-modal semantics (like general categories or style), whereas high-frequency components capture the fine-grained, modality specific details that are crucial for personalization \cite{ong2024spectrum,du2023frequency}. Current models, by fusing signals in the spatial domain, indiscriminately mix these components, leading to over-smoothing where dominant low-frequency signals drown out valuable high-frequency nuances. This lack of spectral awareness points to a deeper information-theoretic failure. An ideal model should not attempt to preserve all information, but rather isolate and retain only the most predictive frequency bands for the recommendation task. Without a principled mechanism to disentangle and filter information, current models create representations that are bloated with redundancy from low-frequency components and corrupted by noise from high-frequency ones. This compromises the model's expressive power and generalization ability, underscoring the need for a new, theoretically-grounded approach.

To address these limitations, we present \textbf{FITMM}, a frequency-aware information-theoretic framework for multimodal recommendation. Our approach is grounded in a spectral view: an orthogonal or tight-frame transform approximately block-diagonalizes bandwise covariances, and the Gaussian Information Bottleneck (IB) decouples across frequency bands. On this foundation, FITMM instantiates a principled ``separate--then--fuse'' pipeline. Concretely, we first build graph-enhanced representations for ID, visual, and textual signals and perform modality-wise spectral decomposition to obtain orthogonal bands. Within each band, we form a lightweight multimodal component to capture cross-modal collaboration while preserving band structure; a residual, task-adaptive gate then aggregates bands into the final representation. To control redundancy and improve generalization, we regularize with a global frequency-domain IB term that allocates capacity across bands (Wiener-style shrinkage of informative directions with shut-off of weak bands). We further introduce a cross-modal spectral consistency loss that aligns modalities. The model is trained end-to-end with a standard objective plus the IB-inspired and consistency terms.

In summary, the main contributions of this work are as follows:
\begin{itemize}
    \item We establish a spectral information-theoretic basis. Under an orthogonal or tight-frame transform that approximately block-diagonalizes bands, the Gaussian IB decouples, supporting a frequency-first separate–then–fuse design.
    \item We instantiate this foundation in FITMM, with modality-wise spectral decomposition, within-band multimodal interaction, residual task-adaptive gating, and a global frequency-domain IB term for capacity allocation.
    \item We conduct extensive experiments on three real-world datasets with standard protocols, demonstrating that FITMM's theoretically grounded design leads to significant performance gains over advanced baselines.
\end{itemize}

\section{RELATED WORK}
\subsection{Multimodal Recommendation}
The increasing availability of multimodal data in platforms such as e-commerce and short video services has notably expanded the scope of recommender systems. To move beyond the limitations of traditional collaborative filtering methods \cite{koren2009matrix,rendle2012bpr}, recent studies have explored incorporating various content modalities, including images, textual descriptions and audio, to better capture user preferences and address data sparsity \cite{xu2023musenet,yu2022graph,han2022modality,yang2023based}. Early approaches often treated multimodal information as auxiliary input and combined it with high-order feature interactions to improve performance \cite{chen2019personalized,deldjoo2021content,he2016vbpr,kang2017visually}. Building on early multimodal fusion strategies, recent research \cite{yang2023multimodal,zhou2023mmrec,geng2023vip5} advances toward fine-grained modeling. To further address the challenges of modality fusion and noise control, several advanced methods have been developed. MMIL \cite{yang2024multimodal} introduces a multi-intention framework that jointly learns intention prototypes and cross-modal semantics for personalized recommendation. AlignRec \cite{liu2024alignrec} improves integration by decomposing the learning objective into alignment tasks at the modality, ID, and user-item levels. SMORE \cite{ong2024spectrum} transforms multimodal signals into the frequency domain to suppress modality-specific noise and incorporates spectral fusion with graph-based modeling. PromptMM \cite{wei2024promptmm} explores a knowledge distillation strategy that combines prompt tuning with modality-aware weighting. In addition, large language models have been widely applied in various downstream tasks\cite{karra2024interarec,chang2025survey,yang2025toward,chen2025tourrank,ping2025hdlcore,li2025climatellm,ye2024domain}. Recent studies have brought notable improvements to recommendations \cite{ye2025harnessing,liu2024rec,xiang2024neural,xia2025hierarchical,xia2025trackrec}.

\subsection{Graph-augmented Recommendation}

Graph-based modeling has become a prominent approach for capturing the relational structure between users and items \cite{tao2020mgat,liu2023multimodal,liu2023megcf}. By exploiting user-item interaction graphs, models such as MMGCN \cite{wei2019mmgcn} and GRCN \cite{wei2020graph} enhance node representations through topology-aware message propagation. LATTICE \cite{zhang2021mining} incorporates higher-order affinities, while FREEDOM \cite{zhou2023tale} adopts denoising strategies to leverage frozen item-item relations for more stable learning. Recent advances further integrate graph learning with self-supervised and contrastive techniques. SGL \cite{wu2021self} applies self-discrimination on user-item bipartite graphs to improve robustness, and MMGCL \cite{yi2022multi} constructs multimodal views using modality-aware augmentations such as edge dropout and masking. SLMRec \cite{tao2022self} and BM3 \cite{zhou2023bootstrap} also employ feature-level augmentations and contrastive objectives to strengthen modality associations. Additional methods such as DualGNN \cite{wang2021dualgnn} and MGCL \cite{liu2023multimodal} integrate multimodal inputs with graph structures and user intention modeling, achieving better adaptability in dynamic content environments. DiffMM \cite{jiang2024diffmm} introduces a diffusion-based multimodal framework that aligns cross-modal signals with collaborative patterns. Together, these graph-enhanced methods emphasize the growing significance of combining structural relations, multimodal content, and contrastive learning for more robust and expressive recommendation.



\section{The Proposed Model}
In this section, we give a detailed introduction to the proposed method. The overall framework of the model is shown in Fig.~\ref{overall_model}.

\begin{figure*}
  \centering
  \includegraphics[width=\linewidth]{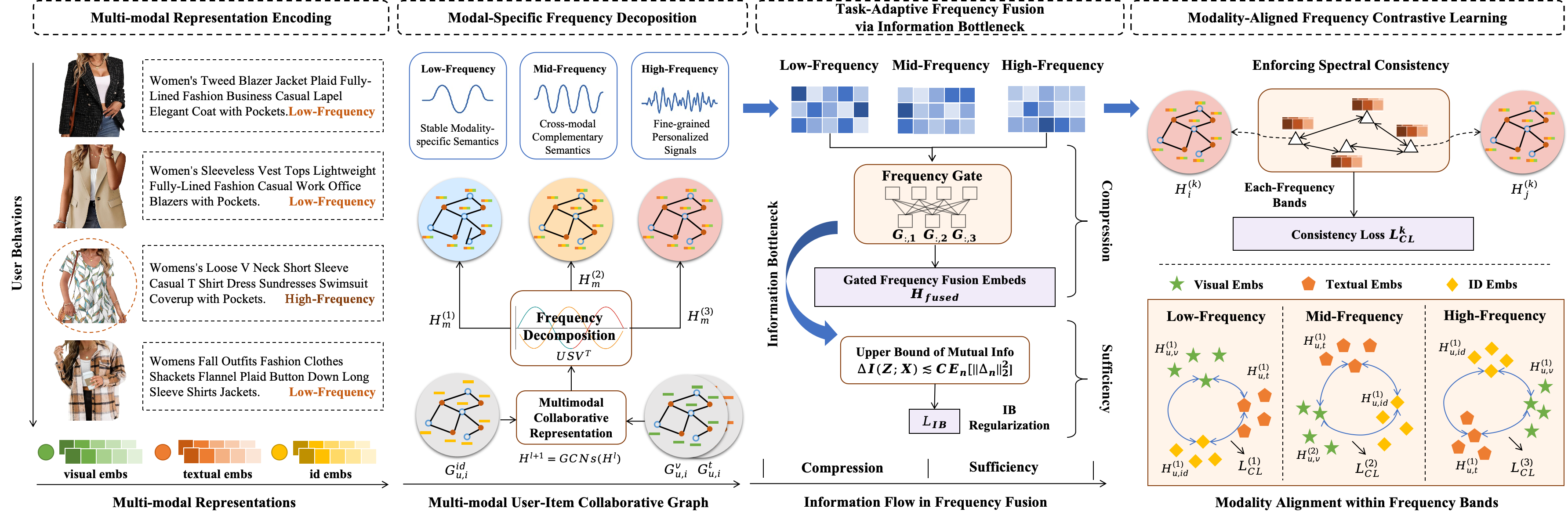}
  \caption{The overall architecture of the proposed FITMM framework, which integrates modality-specific frequency decomposition, task-adaptive fusion, and information-theoretic regularization to enable fine-grained and robust recommendation.}

  \Description{Different Attention Layer size.}
  \label{overall_model}
\end{figure*}

\subsection{A Spectral Information-Theoretic View of Multimodal Fusion}
\label{sec:spectral_foundation}

In multimodal recommendation, items are typically described using heterogeneous signals such as images and textual descriptions. These signals jointly capture the underlying user–item interaction patterns. Existing methods largely follow the classical paradigms of early or late fusion through concatenation or gating mechanisms \cite{wei2019mmgcn,zhou2023tale}. This implicitly assumes semantic information is uniformly distributed across the representation space. However, we argue that this assumption is flawed, as empirical evidence suggests otherwise:

\begin{itemize}
    \item \textbf{Cross-modal Collaborative Semantics} (e.g., visual style, coarse categories, general sentiment), which are essential for recommendation, tend to reside in stable, \emph{low-frequency} components of the signal.
    \item \textbf{Modality-specific Details and Noise} (e.g., fine-grained visual textures, specific writing styles, or background noise) often lie in more volatile \emph{mid- to high-frequency} bands and vary significantly across modalities.
\end{itemize}

Fusing signals in the spatial domain \emph{before} disentangling these components forces the model to learn \emph{both} (a) how to extract the shared collaborative signal and (b) how to suppress modality-specific noise \emph{simultaneously}. This entanglement increases optimization complexity and risks diluting the core collaborative signal that drives effective recommendations.

To ground this design beyond heuristics, we pose a first-principles question:

\begin{quote}
\itshape
\textbf{Under what structural and capacity assumptions does a frequency-domain ``separate–then–fuse'' strategy provably retain at least as much \emph{collaborative} information as direct spatial fusion?}
\end{quote}


\paragraph{From filter banks to information bottlenecks.}
Classical multiresolution analysis and perfect-reconstruction filter banks (wavelets, Laplacian pyramids) show that signals can be stably decomposed into orthogonal (or tight-frame) subbands and then processed and recombined with \emph{perfect (or near-perfect) reconstruction under standard PR conditions} \cite{vetterli1995wavelets,mallat1999wavelet,burt1987laplacian}. In parallel, Wiener theory shows that in additive signal+noise models, the MMSE estimator is \emph{equivalently a spatial convolution whose frequency response performs per-band SNR weighting}, rather than a single \emph{band-agnostic} spatial weight \cite{kailath1981lectures,pratt2006generalized,robinson1967principles,chen2006new,gardner2002cyclic}. Together, these suggest that a decomposition\,$\rightarrow$\,per-band processing\,$\rightarrow$\,fusion pipeline is a principled way to preserve informative structure while attenuating nuisance variability.

\begin{proposition}[Bandwise Gaussian Information Bottleneck (GIB): existence of a band-separable optimum]
\label{prop:gib_bandwise}
Let $X=\mathrm{concat}(X^{(1)},\dots,X^{(K)})$ be modality-aware frequency components obtained by an orthogonal (or tight-frame) decomposition of multimodal embeddings (e.g., SVD/wavelet/graph-wavelet). Suppose $(X,Y)$ are jointly Gaussian, and \emph{both} $\Sigma_{XX}\!=\!\mathrm{Cov}(X)$ and $\Sigma_{X|Y}\!=\!\mathrm{Cov}(X|Y)$ are \emph{block-diagonal in the same band basis}. Consider the Information Bottleneck Lagrangian $\mathcal{L}_{\mathrm{IB}}(p(Z|X))=I(Z;X)-\beta I(Z;Y)$. Then:
\begin{enumerate}
    \item There \emph{exists} an optimal encoder that respects the band partition (i.e., $Z=\bigoplus_k Z^{(k)}$ with $p(Z|X)=\prod_k p(Z^{(k)}|X^{(k)})$), for which the objective \emph{equals the sum} of bandwise objectives: $\mathcal{L}_{\mathrm{IB}}=\sum_{k=1}^K \mathcal{L}^{(k)}_{\mathrm{IB}}$.
    \item For each band, an optimal encoder is a (noisy) \emph{linear} projection onto the eigenvectors of the generalized eigenproblem associated with the pair $(\Sigma_{XX},\Sigma_{X|Y})$; \emph{equivalently}, onto the eigenvectors of the CCA operator $\Sigma_{XX}^{-1}\Sigma_{XY}\Sigma_{YY}^{-1}\Sigma_{YX}$, yielding a \emph{band-specific shrinkage/gating} that is monotone in a relevance-to-redundancy tradeoff \cite{tishby2000information,chechik2003information}.
\end{enumerate}
\end{proposition}

\noindent\textit{Proof sketch.}
For Gaussians, $I(Z;X)$ and $I(Z;Y)$ reduce to $\log\det$ expressions in covariances. Under block-diagonality \emph{in the same band basis}, any band-respecting encoder yields additive mutual informations over bands; by standard convexity/orthogonality arguments for the Gaussian IB, an optimal solution can be chosen within this band-respecting class. The per-band minimizer is the Gaussian-IB linear–noisy projection aligned with the stated eigenproblems (closely related to CCA) with coefficients controlled by $\beta$ \cite{tishby2000information,chechik2003information}, unique up to within-band rotations.

\begin{corollary}[Connection to Wiener weighting]
\label{cor:wiener}
In the special case where $Y=S$ (or any invertible linear observation of $S$ on the relevant subspace) and $X^{(k)}=S^{(k)}+N^{(k)}$ with independent WSS signal/noise per band, the bandwise GIB encoder reduces to a linear shrinkage proportional to per-band SNR, recovering the classical Wiener form
$G^{(k)}(\omega)=\tfrac{\Phi_{SS}^{(k)}(\omega)}{\Phi_{SS}^{(k)}(\omega)+\Phi_{NN}^{(k)}(\omega)}$ \emph{up to a $\beta$-dependent scaling}. As $\beta\!\to\!\infty$ it approaches the pure MMSE/Wiener solution; for finite $\beta$ it yields SNR-dependent shrinkage.
Thus, a \emph{learnable} gate per frequency band implements a data-driven Wiener-like relevance filter \cite{robinson1967principles,chen2006new,gardner2002cyclic}.
\end{corollary}

\paragraph{Why frequency-aware fusion helps GNNs.}
On graphs, message passing corresponds to low-pass filtering in the Laplacian spectrum; deep stacks tend toward over-smoothing (loss of high-frequency or heterophilous cues) \cite{shuman2013emerging,ortega2018graph,sandryhaila2013discrete,oono2019graph,xu2019graph,hammond2011wavelets}. An explicit spectral decomposition (graph wavelets / spectral filters) with \emph{bandwise} fusion mitigates this by preserving both low-frequency collaborative semantics and mid-/high-frequency, modality-specific details.

\paragraph{Compatibility with deep networks’ spectral bias.}
Deep nets preferentially learn low frequencies; high frequencies are learned late and are fragile \cite{rahaman2019spectral}. Providing explicit frequency channels (and even Fourier features when appropriate) counteracts this bias and improves modeling of fine detail \cite{tancik2020fourier}. This aligns with our design: disentangle into frequency bands, apply learnable, IB-regularized gates, and then fuse.



\subsection{Disentangled Spectral Representation and Adaptive Fusion}

Building on the theoretical foundation of frequency-based information preservation, we now present a practical framework that operationalizes the principle of \emph{“separate first, fuse later”}. Specifically, we design a representation pipeline that first disentangles spectral components across modalities to isolate shared semantics from modality-specific residuals, and then performs adaptive fusion guided by task relevance. The pipeline consists of three key stages: structural modeling via item–user graphs, spectral decomposition of modality inputs, and frequency-aware fusion tailored to downstream objectives.

\subsubsection{Graph-based Multimodal Representation and Decomposition}

Given a user-item interaction graph $\mathcal{G} = (\mathcal{U} \cup \mathcal{I}, \mathcal{E})$, we employ independent Graph Convolutional Networks (GCNs) to preserve modality-specific semantics and topological patterns. The layer-wise propagation rule follows:
\begin{equation}
H^{(l+1)} = \sigma(\tilde{A} H^{(l)} W^{(l)}), \quad H^{(0)} = E,
\end{equation}
where $E \in \mathbb{R}^{N \times d}$ is the input feature matrix, $\tilde{A}  \in \mathbb{R}^{N \times N}$ is the normalized adjacency matrix of the interaction graph, $W^{(l)}  \in \mathbb{R}^{d \times d}$ is a trainable transformation, and $\sigma(\cdot)$ is a non-linear activation. This is applied to $H_{\text{id}}, H_{\text{v}}, H_{\text{t}} \in \mathbb{R}^{N \times d}$ to produce modality-specific representations for users and items. To enhance high-order reasoning, we optionally integrate user-user and item-item graphs, constructed via collaborative similarity or content similarity.

To further disentangle multi-scale semantics, we apply frequency decomposition on each modality through Singular Value Decomposition, enabling an interpretable and orthogonal projection into frequency subspaces:
\begin{equation}
H_m = U S V^\top = \sum_{k=1}^{K} U S^{(k)} V^\top = \sum_{k=1}^{K} H_m^{(k)},
\end{equation}
where $S^{(k)}$ denotes the $k$-th frequency band after partitioning the diagonal matrix $S$ by magnitude. Each $H_m^{(k)}$ thus represents the latent structure encoded in the $k$-th frequency range.

We repeat this decomposition independently for ID, visual, and textual embeddings. Then, at each frequency band $k$, we construct multimodal representations by concatenation:
\begin{equation}
H^{(k)} = \text{concat}(H_{\text{id}}^{(k)}, H_{\text{v}}^{(k)}, H_{\text{t}}^{(k)}).
\end{equation}
This allows the model to preserve modality-specific semantics while enabling downstream frequency-aware operations.

\subsubsection{Task-Adaptive Frequency Fusion}

Not all frequency components contribute equally to the recommendation task: low-frequency bands typically encode general semantics, while high-frequency bands carry fine-grained personalization cues. To selectively retain task-relevant signals, we design a gated frequency fusion mechanism that dynamically reweights each frequency component based on its informativeness. We first compute the gating vector $G \in \mathbb{R}^{K}$ from the original task representation $H$:
\begin{equation}
G = \sigma(W_g H + b_g),
\end{equation}
where $W_g$ and $b_g$ are learnable parameters and $\sigma(\cdot)$ is a sigmoid activation to ensure $G_k \in (0,1)$.

The final representation is obtained as a soft aggregation over frequency-specific features:
\begin{equation}
H_{\text{fused}} = \sum_{k=1}^{K} G_k \cdot H^{(k)}.
\end{equation}

Unlike traditional late fusion or modality attention, our method provides a fine-grained frequency selection mechanism, enabling the model to retain only semantically and structurally aligned information. This formulation underpins the subsequent frequency-domain IB regularization, with the gate controlling a global capacity budget that drives reverse water-filling across bands.

\subsection{Incremental Frequency-Domain Information Bottleneck}
\label{sec:freq_ib_regularization}

While \S\ref{sec:spectral_foundation} shows that (under approximate block-diagonal covariances in the \emph{same} band basis) the Gaussian-IB problem \emph{decouples at optimum} (i.e., it admits a band-respecting optimal encoder whose optimal value equals the sum of per-band optima), here we study how a \emph{global frequency-domain IB budget} allocates capacity across bands and why the induced gates improve generalization.

\paragraph{Bandwise IB with a global budget.}
Let \(H=\bigoplus_{k=1}^{K} H^{(k)}\) be the bandwise decomposition (orthonormal or tight-frame).
We form pre-fusion variables
\begin{equation}
\label{eq:band-gated-encoder}
\tilde Z^{(k)} = G^{(k)} \odot H^{(k)} + \xi^{(k)}, \qquad
Z = A\Big(\bigoplus_{k=1}^{K} \tilde Z^{(k)}\Big).
\end{equation}
\noindent with encoder noise \(\xi^{(k)}\) and a band-respecting isometry \(A\).

\paragraph{Identity-preserving gating} We use \(G^{(k)}=\mathbf 1+\Delta^{(k)}\) with \(\Delta^{(k)}\) initialized near zero, and control \emph{incremental} informations relative to the identity baseline \(\tilde Z^{(k)}_0\) (where \(\Delta^{(k)}=0\)):
$
\Delta I_k^{\mathrm{in}}:=I(\tilde Z^{(k)};H^{(k)})-I(\tilde Z^{(k)}_0;H^{(k)})$,
$\Delta I_k^{\mathrm{rel}}:=I(\tilde Z^{(k)};Y)-I(\tilde Z^{(k)}_0;Y).$
We allocate a single \emph{global} budget on the increments:
\begin{equation}
\max_{\{p(\tilde Z^{(k)}|H^{(k)})\}} \;\sum_{k=1}^{K} \Delta I_k^{\mathrm{rel}}
\quad\text{s.t.}\quad
\sum_{k=1}^{K} \Delta I_k^{\mathrm{in}} \;\le\; \mathcal{B}.
\label{eq:freq_ib_budget_incremental}
\end{equation}
Under the same Gaussian/block-diagonal assumptions, \emph{decoupling at optimum} and KKT/\emph{reverse water-filling} remain valid for~\eqref{eq:freq_ib_budget_incremental}. Moreover, for a scalar direction in band \(k\) with gate \(1+\delta\) and \(\gamma_k:=\mathrm{Var}(H^{(k)})/\mathrm{Var}(\xi^{(k)})\),
\begin{equation}
\Delta I_k^{\mathrm{in}}
\;=\;\tfrac12 \log\!\frac{1+\gamma_k(1+\delta)^2}{\,1+\gamma_k\,}
\;\le\; c_k\,\delta^2,\qquad c_k:=\frac{\gamma_k}{1+\gamma_k},
\end{equation}
and analogously \(\sum_k \Delta I_k^{\mathrm{in}}\lesssim \|\Delta\|_2^2\).

\begin{proposition}[Reverse water-filling for frequency-IB]
\label{prop:ib_waterfilling}
Assume a Gaussian band model where \((H,Y)\) are jointly Gaussian and bandwise covariance blocks are (approximately) block-diagonal in the same band basis. Within each band, the Gaussian-IB optimal encoders are noisy linear projections aligned with a generalized eigenproblem for \((\Sigma_{HH},\Sigma_{H|Y})\), equivalently the CCA operator \(\Sigma_{HH}^{-1}\Sigma_{HY}\Sigma_{YY}^{-1}\Sigma_{YH}\) \cite{chechik2003information}.
Then the constrained program~\eqref{eq:freq_ib_budget_incremental} satisfies KKT conditions with \emph{reverse water-filling}:
there exists \(\lambda>0\) s.t.
\[
\frac{\partial \Delta I_k^{\mathrm{rel}}}{\partial \Delta I_k^{\mathrm{in}}}=\lambda \ \text{for active bands,}\qquad
\frac{\partial \Delta I_k^{\mathrm{rel}}}{\partial \Delta I_k^{\mathrm{in}}}<\lambda \ \text{for inactive bands.}
\]
In particular, along active eigen-directions the optimal gate is a \emph{Wiener-like} shrinkage (up to a $\beta$-dependent scaling), with shut-off below a data-dependent threshold; when \(Y=S\) (or an invertible linear observation of \(S\)) and \(X^{(k)}=S^{(k)}\!+\!N^{(k)}\) (independent WSS), letting \(\beta\!\to\!\infty\) recovers MMSE/Wiener.
\end{proposition}

\noindent\textit{Sketch.}
Gaussian-IB converts the informations to log-determinants on bandwise blocks and yields the stated linear encoders; a single budget on \(\sum_k \Delta I_k^{\mathrm{in}}\) equalizes marginal gain across active bands, mirroring Gaussian rate–distortion \cite{cover1999elements,ash2012information}.

\paragraph{Generalization via information control.}
Under standard conditions (bounded/sub-Gaussian losses, randomized encoders, i.i.d.\ samples),
\(
|\mathrm{gen}|\lesssim \sqrt{\tfrac{2}{n}I(Z;H)}=\sqrt{\tfrac{2}{n}\sum_k I_k^{\mathrm{in}}}
\)
\cite{xu2017information,bu2020tightening}.
Since \(I_k^{\mathrm{in}}=I(\tilde Z^{(k)}_0;H^{(k)})+\Delta I_k^{\mathrm{in}}\), controlling \(\sum_k \Delta I_k^{\mathrm{in}}\) controls the \emph{additional} information beyond the baseline, giving an explicit \(\mathcal{O}(\sqrt{\mathcal{B}/n})\) bound.

\paragraph{Operational surrogate (global coupling).}
For the budget surrogate we use one scalar per band \(G_{n,k}=1+\Delta_{n,k}\). Using the bound above and \(\log(1+x)\le x\),
$
\sum_k \Delta I^{\mathrm{in}}_{n,k}\ \lesssim\ \|\Delta_n\|_2^2,\qquad \Delta_n=(\Delta_{n,1},\ldots,\Delta_{n,K}).
$
We adopt the nonnegative, globally coupled surrogate on increments:
\begin{equation}
\label{eq:our_ib_surrogate_increment}
\mathcal{L}_{\mathrm{IB}}
\;=\;
\alpha\,\tfrac{1}{N}\sum_{n}\|\Delta_n\|_2^2
\;+\;
\mu\,\tfrac{1}{N}\sum_{n}\|\Delta_n\|_2\;\Big(\sum_{k}[\,\Delta_{n,k}-\phi_{+}\,]_+\Big),
\end{equation}
with \([x]_+=\max(x,0)\) and \(\alpha,\mu>0\).
The first term upper-bounds the incremental IB budget; the second induces cross-band competition via the shared factor \(\|\Delta_n\|_2\) (softplus is a smooth alternative).

\subsection{Auxiliary Regularization: Enforcing Spectral Consistency}
\label{sec:contrastive_learning}

Our framework is grounded in a spectral information-theoretic basis in which an orthogonal or tight-frame transform approximately block-diagonalizes bandwise covariances and the Gaussian Information Bottleneck decouples across frequency bands. To further strengthen the model's adherence to our first principle, we introduce a targeted auxiliary objective. This objective, termed Cross-Modal Spectral Consistency, ensures that for a given item, the representations from different modalities are explicitly aligned within each corresponding frequency band. 

Let $H^{(k)} = [H^{(k)}_{\text{id}} \| H^{(k)}_{\text{vis}} \| H^{(k)}_{\text{txt}}]$ be the concatenated representation for all items in the $k$-th frequency band, decomposed into its constituent modal parts. We enforce alignment by minimizing a squared cosine distance loss between all pairs of modalities:

\begin{equation}
\mathcal{L}_{\text{CL}} = \sum_{k=1}^{K} \sum_{u=1}^{N} \sum_{(i, j) \in \mathcal{P}} \left( 1 - \cos \left( H^{(k)}_{u,i}, H^{(k)}_{u,j} \right) \right)^2
\label{eq:consistency_loss}
\end{equation}
where $\mathcal{P}$ is the set of modality pairs, i.e., $\{(\text{id}, \text{vis}), (\text{id}, \text{txt}), (\text{vis}, \text{txt})\}$. This loss encourages the different modal views of an item to map to a consistent point in the embedding space for each frequency scale.

\subsection{Training and Optimization}
Our model aims to predict user-item interaction probabilities by leveraging historical behaviors and multimodal content. For each user $u$ and item $i$, the final representations $\mathbf{h}_u$ and $\mathbf{h}_i$ are obtained through frequency-adaptive learning, and the interaction score is computed as $\hat{y}_{ui} = \mathbf{h}_u^\top \mathbf{h}_i$. The model is trained using the binary cross-entropy loss:
\begin{equation}
\mathcal{L}_{\text{BCE}} = - \sum_{(u, i) \in \mathcal{D}} \left[ y_{ui} \log \hat{y}_{ui} + (1 - y_{ui}) \log (1 - \hat{y}_{ui}) \right],
\end{equation}
where $y_{ui} \in \{0,1\}$ denotes the label and $\mathcal{D}$ is the training set.

The final objective combines the recommendation loss with two regularization terms: the information bottleneck loss and the contrastive loss:
\begin{equation}
\mathcal{L} = \mathcal{L}_{\text{BCE}} + \lambda \mathcal{L}_{\text{IB}} + \eta \mathcal{L}_{\text{CL}},
\end{equation}
where $\lambda$ and $\eta$ control the strength of each regularizer.

\section{EXPERIMENTS}
In this section, we conduct extensive experiments to address the following research questions:

\begin{itemize}
    \item \textbf{RQ1:} - How does our proposed method perform compared to state-of-the-art (SOTA) multimodal baselines?
    \item \textbf{RQ2:} - Does our approach exhibit superior performance in cold-start scenarios compared to existing advanced models?
    \item \textbf{RQ3:} - How do frequency modeling and bottleneck regularization affect performance?
    \item \textbf{RQ4:} – What roles do different frequency components play in capturing modality-specific and shared semantics?
    \item \textbf{RQ5:} – How does the model improve representation quality for the cold-start prediction?
    \item \textbf{RQ6:} - How does the model behave under different hyperparameter settings?
\end{itemize}

\begin{table*}[t]
\centering
\caption{
Overall recommendation performance on three Amazon datasets. Our proposed FITMM consistently outperforms all baselines across all metrics, demonstrating its superior capability in modeling multimodal signals.
}
\label{tab:main_results}
\resizebox{\textwidth}{!}{
\begin{tabular}{lcccccccccccccccc}
\toprule
\textbf{Dataset} & \textbf{Metric} & BPR & LightGCN & VBPR & MMGCN & GRCN & DualGNN & SLMRec & LATTICE & BM3 & FREEDOM & DiffMM & MMIL & AlignRec & SMORE & \textbf{FITMM} \\
\midrule
\multirow{4}{*}{Baby} 
& R@10  & 0.0357 & 0.0479 & 0.0418 & 0.0413 & 0.0538 & 0.0507 & 0.0533 & 0.0561 & 0.0573 & 0.0624 & 0.0617 & 0.0670 & 0.0674 & \underline{0.0680} & \textbf{0.0716} \\
& R@20  & 0.0575 & 0.0754 & 0.0667 & 0.0649 & 0.0832 & 0.0782 & 0.0788 & 0.0867 & 0.0904 & 0.0985 & 0.0978 & 0.1035 & \underline{0.1046} & 0.1035 & \textbf{0.1089} \\
& N@10  & 0.0192 & 0.0257 & 0.0223 & 0.0211 & 0.0285 & 0.0264 & 0.0291 & 0.0305 & 0.0311 & 0.0324 & 0.0321 & 0.0361 & 0.0363 & \underline{0.0365} & \textbf{0.0387} \\
& N@20  & 0.0249 & 0.0328 & 0.0287 & 0.0275 & 0.0364 & 0.0335 & 0.0365 & 0.0383 & 0.0395 & 0.0416 & 0.0408 & 0.0455 & \underline{0.0458} & 0.0457 & \textbf{0.0484} \\
\midrule
\multirow{4}{*}{Sports} 
& R@10  & 0.0432 & 0.0569 & 0.0561 & 0.0394 & 0.0607 & 0.0574 & 0.0667 & 0.0628 & 0.0659 & 0.0705 & 0.0687 & 0.0747 & 0.0758 & \underline{0.0762} & \textbf{0.0809} \\
& R@20  & 0.0653 & 0.0864 & 0.0857 & 0.0625 & 0.0928 & 0.0881 & 0.0998 & 0.0961 & 0.0987 & 0.1077 & 0.1035 & 0.1133 & \underline{0.1160} & 0.1142 & \textbf{0.1187} \\
& N@10  & 0.0241 & 0.0311 & 0.0305 & 0.0203 & 0.0335 & 0.0316 & 0.0366 & 0.0339 & 0.0357 & 0.0382 & 0.0357 & 0.0405 & \underline{0.0414} & 0.0408 & \textbf{0.0441} \\
& N@20  & 0.0298 & 0.0387 & 0.0386 & 0.0266 & 0.0421 & 0.0393 & 0.0454 & 0.0431 & 0.0443 & 0.0478 & 0.0458 & 0.0505 & \underline{0.0517} & 0.0506 & \textbf{0.0538} \\
\midrule
\multirow{4}{*}{Clothing} 
& R@10  & 0.0206 & 0.0361 & 0.0283 & 0.0221 & 0.0425 & 0.0447 & 0.0440 & 0.0503 & 0.0577 & 0.0616 & 0.0593 & 0.0643 & 0.0651 & \underline{0.0659} & \textbf{0.0698} \\
& R@20  & 0.0303 & 0.0544 & 0.0418 & 0.0357 & 0.0661 & 0.0663 & 0.0655 & 0.0755 & 0.0845 & 0.0917 & 0.0874 & 0.0961 & \underline{0.0993} & 0.0987 & \textbf{0.1017} \\
& N@10  & 0.0114 & 0.0197 & 0.0162 & 0.0116 & 0.0227 & 0.0237 & 0.0238 & 0.0277 & 0.0316 & 0.0333 & 0.0325 & 0.0348 & 0.0356 & \underline{0.0360} & \textbf{0.0378} \\
& N@20  & 0.0138 & 0.0243 & 0.0196 & 0.0151 & 0.0283 & 0.0289 & 0.0293 & 0.0356 & 0.0387 & 0.0409 & 0.0396 & 0.0428 & 0.0437 & \underline{0.0443} & \textbf{0.0457} \\
\bottomrule
\end{tabular}
}
\end{table*}

\subsection{Experimental Settings}

\subsubsection{Datasets}

We conduct experiments on three categories of the widely-used Amazon \footnote{\url{http://jmcauley.ucsd.edu/data/amazon/}} Review dataset~\cite{lakkaraju2013s}: \textit{Baby}, \textit{Sports and Outdoors}, and \textit{Clothing, Shoes and Jewelry}, abbreviated as Baby, Sports, and Clothing. Following common practice in FREEDOM~\cite{han2022modality} and SMORE\cite{ong2024spectrum}, we apply the 5-core filtering and utilize the same multimodal inputs. Dataset statistics are summarized in Table~\ref{tab:datasets Statistics}.

\begin{table}
  \caption{ Statistics of the experimental Amazon datasets.}
  \label{tab:datasets Statistics}
  \begin{tabular}{ccccccc}
    \toprule
    Dataset & User & Item & Interactions & Density \\
    \midrule
    Baby & 19,445 & 7,050 & 139,110 & 0.101\%\\
    Sports & 35,598 & 18,357 & 256,308 & 0.039\%\\
    Clothing & 39,387 & 23,033 & 237,488 & 0.026\%\\
  \bottomrule
\end{tabular}
\end{table}

\subsubsection{Baselines}
To evaluate the performance, we compared with the following baselines. Traditional collaborative filtering methods: \textbf{BPR} \cite{rendle2012bpr} and \textbf{LightGCN} \cite{he2020lightgcn}. Recent multimodal recommendation methods:  \textbf{VBPR} \cite{he2016vbpr}, \textbf{MMGCN} \cite{wei2019mmgcn}, \textbf{GRCN} \cite{wei2020graph}, \textbf{DualGNN} \cite{wang2021dualgnn}, \textbf{SLMRec} \cite{tao2022self}, \textbf{LATTICE} \cite{zhang2021mining}, \textbf{BM3} \cite{zhou2023bootstrap}, \textbf{FREEDOM} \cite{zhou2023tale}, \textbf{DiffMM} \cite{jiang2024diffmm}, \textbf{MMIL} \cite{yang2024multimodal}, \textbf{AlignRec} \cite{liu2024alignrec}, \textbf{SMORE} \cite{ong2024spectrum}.

\subsubsection{Evaluation Protocols}
We adopt the all-ranking protocol for top-K recommendation, where all candidate items are ranked for each user. Following standard practice~\cite{mu2022learning}, we report Recall@K and NDCG@K for K in \{10, 20\}, which measure retrieval accuracy and ranking quality, respectively. User interaction data is split into 80\% training, 10\% validation, and 10\% testing. During training, negative sampling is applied by pairing each observed interaction with randomly sampled negative items. 

\subsubsection{Implementation Details}
Our method is implemented in PyTorch using the MMRec framework~\cite{zhou2023mmrec}. We set the embedding dimension to 64 and initialize parameters with Xavier initialization~\cite{glorot2010understanding}. The model is optimized using Adam~\cite{kingma2014adam}, with the learning rate tuned in $\{0.0001, 0.0005, 0.001, 0.005\}$ and a batch size of 2048. Training runs with early stopping based on Recall@20, triggered after 20 validation steps without improvement. We tune the information bottleneck weight $\lambda$ and contrastive loss weight $\eta$ in $\{0.0001, 0.001, 0.01, 0.1, 1.0\}$, and vary the number of frequency bands $M$ in $\{2, 3, 4, 5\}$ to evaluate decomposition effectiveness. All experiments are conducted on a NVIDIA A100 GPU.

\begin{table*}
\small
\centering
  \caption{
Performance comparison under the cold-start scenario, where users with $\leq$ 5 interactions are used for experiment. FITMM achieves notable improvements over strong baselines, confirming its robustness and adaptability in data-sparse settings.
}

  \label{tab:cold_performance}
  \begin{tabular}{cccccccccccccccc}
    \toprule
    \multirow{2}{*}{Model} & \multicolumn{4}{c}{\textbf{Amazon-Baby}} & \multicolumn{4}{c}{\textbf{Amazon-Sports}} & \multicolumn{4}{c}{\textbf{Amazon-Clothing}}\\ 
    & R@10 & R@20 & N@10 & N@20 & R@10 & R@20 & N@10 & N@20 & R@10 & R@20 & N@10 & N@20\\
    \midrule
GRCN & 0.0510 & 0.0770 & 0.0268 & 0.0333 & 0.0556 & 0.0814 & 0.0305 & 0.0370 & 0.0416 & 0.0648 & 0.0213 & 0.0271 \\
BM3  & 0.0549 & 0.0845 & 0.0288 & 0.0363 & 0.0581 & 0.0920 & 0.0305 & 0.0390 & 0.0424 & 0.0614 & 0.0229 & 0.0276 \\
LATTICE & 0.0570 & 0.0827 & 0.0298 & 0.0362 & 0.0380 & 0.0570 & 0.0205 & 0.0252 & 0.0414 & 0.0555 & 0.0218 & 0.0254 \\
FREEDOM & 0.0535 & 0.0880 & 0.0297 & 0.0384 & 0.0622 & 0.0933 & 0.0328 & 0.0406 & 0.0444 & 0.0671 & 0.0241 & 0.0298 \\
MMIL & 0.0672 & 0.1014 & 0.0377 & 0.0463 & 0.0759 & 0.1141 & \underline{0.0422} & \underline{0.0517} & 0.0635 & 0.0948 & 0.0341 & 0.0423 \\
SMORE & \underline{0.0687} & \underline{0.1019} & \underline{0.0381} & \underline{0.0464} & \underline{0.0788} & \underline{0.1158} & 0.0420 & 0.0512 & \underline{0.0676} & \underline{0.0983} & \underline{0.0369} & \underline{0.0446} \\
\hline
FITMM & \textbf{0.0752} & \textbf{0.1109} & \textbf{0.0413} & \textbf{0.0503} & \textbf{0.0820} & \textbf{0.1218} & \textbf{0.0435} & \textbf{0.0535} & \textbf{0.0691} & \textbf{0.1003} & \textbf{0.0381} & \textbf{0.0459} \\
    \bottomrule
  \end{tabular}
\end{table*}

\subsection{Overall Performance (RQ1)}
Table~\ref{tab:main_results} compares performance across three Amazon datasets. Our method consistently outperforms all state-of-the-art (SOTA) baselines on both Recall and NDCG. Specifically, it improves R@10 from 0.0680 to 0.0716 on Baby, from 0.0762 to 0.0809 on Sports, and from 0.0659 to 0.0698 on Clothing. These consistent gains across domains demonstrate the strength of our frequency-aware design. By decomposing multimodal representations into frequency bands and applying task-adaptive fusion guided by information bottleneck principles, our model effectively highlights informative components while suppressing noise, particularly in high-frequency ranges. Compared with MMIL and AlignRec, which enhance modeling through multi-intention learning or hierarchical alignment, our approach introduces a structured frequency-domain perspective. While SMORE also employs frequency-based denoising, it relies on global FFT filtering and a unified complex-valued filter, which limits interpretability and flexibility. In contrast, our method adopts component-wise decomposition and fine-grained, modality-aware fusion, further strengthened by contrastive and bottleneck regularization. These results confirm that our approach not only boosts accuracy but also provides a principled and interpretable framework for multimodal recommendation.

\subsection{Cold-Start Recommendation (RQ2)}
To evaluate performance under data sparsity, we conduct experiments on cold-start users who have no more than five interactions. As shown in Table~\ref{tab:cold_performance}, our method consistently outperforms all representative baselines across the three datasets. Interestingly, it even achieves better results on cold-start users than on the full user set, indicating strong generalization capability. This result reflects the effectiveness of our frequency decomposition and information bottleneck design, which enable the model to extract stable and task-relevant signals from multimodal content. The spectral energy analysis in Fig.~\ref{feq_band_w} further reveals that cold-start items tend to concentrate information in low-frequency bands. This observation aligns with the model’s structural preference for capturing general semantics. In contrast, methods such as FREEDOM and LATTICE, which rely heavily on collaborative graph structures, show significant performance drops under cold-start conditions. Although SMORE and MMIL remain competitive by using frequency filtering and modality-aware intent modeling, our method provides a more comprehensive solution by jointly leveraging frequency semantics, modality structure, and task relevance. Through selective fusion across frequency components, the model offers both robustness and interpretability in cold-start recommendation settings.

\subsection{Ablation Study (RQ3)}

We conduct ablation studies to assess the contribution of each key module in our framework. Table~\ref{tab:exp_ablation} summarizes the performance of several variants: (1) \textit{r/p AF}, which replaces the adaptive frequency fusion with naive averaging; (2) \textit{w/o AF}, which removes frequency decomposition and fusion entirely; (3) \textit{w/o IB}, which discards the information bottleneck loss; (4) \textit{w/o CL}, which removes contrastive learning; and (5) \textit{w/o MM}, which eliminates all multimodal inputs. Among them, \textit{w/o AF} shows the most significant performance drop, demonstrating that frequency modeling is central to capturing multi-scale user preferences. Even the simpler variant \textit{r/p AF}, which retains frequency components but lacks task awareness, results in substantial degradation (e.g., Baby R@10 drops from 0.0716 to 0.0646). These observations highlight the necessity of adaptive, task-aware fusion to prioritize semantically relevant frequency bands and suppress noise. Furthermore, removing the information bottleneck (\textit{w/o IB}) leads to consistent drops across all datasets, verifying its role in filtering redundant components and enforcing informative compression. While the impact of removing contrastive learning (\textit{w/o CL}) is less severe, this module remains important for enhancing modality alignment, which supports better personalization. Finally, removing multimodal inputs (\textit{w/o MM}) causes the steepest performance decline, emphasizing that visual and textual modalities are not only semantically rich but also provide essential frequency-diverse cues for fine-grained preference modeling. 

\begin{table}
\centering
\small
  \label{tab:result}
  \begin{tabular}{ccccccc}
      \toprule
      \textbf{Model} & \multicolumn{2}{c}{\textbf{Baby}} & \multicolumn{2}{c}{\textbf{Sports}} & \multicolumn{2}{c}{\textbf{Clothing}}\\
      & R10 & N10 & R10 & N10 & R10 & N10 \\
    \hline
    r/p AF & 0.0646 & 0.0341 & 0.0767 & 0.0421 & 0.0659 & 0.0356 \\

    w/o AF & 0.0554 & 0.0305 & 0.0605 & 0.0329 & 0.0402 & 0.0219 \\

    w/o IB & 0.0702 & 0.0378 & 0.0785 & 0.0424 & 0.0681 & 0.0364 \\

    w/o CL & 0.0708 & 0.0383 & 0.0791 & 0.0431 & 0.0688 & 0.0369 \\

    w/o MM & 0.0517 & 0.0280 & 0.0539 & 0.0296 & 0.0373 & 0.0200 \\

    \hline
    FITMM & \textbf{0.0716} & \textbf{0.0387} & \textbf{0.0809} & \textbf{0.0441} & \textbf{0.0698} & \textbf{0.0378} \\

    \bottomrule
  \end{tabular}
  \caption{
Ablation study of key components in FITMM. Each variant modifies one module to evaluate its contribution. }

  \label{tab:exp_ablation}
\end{table}

\begin{figure*}[t]
  \centering
  \includegraphics[width=\linewidth]{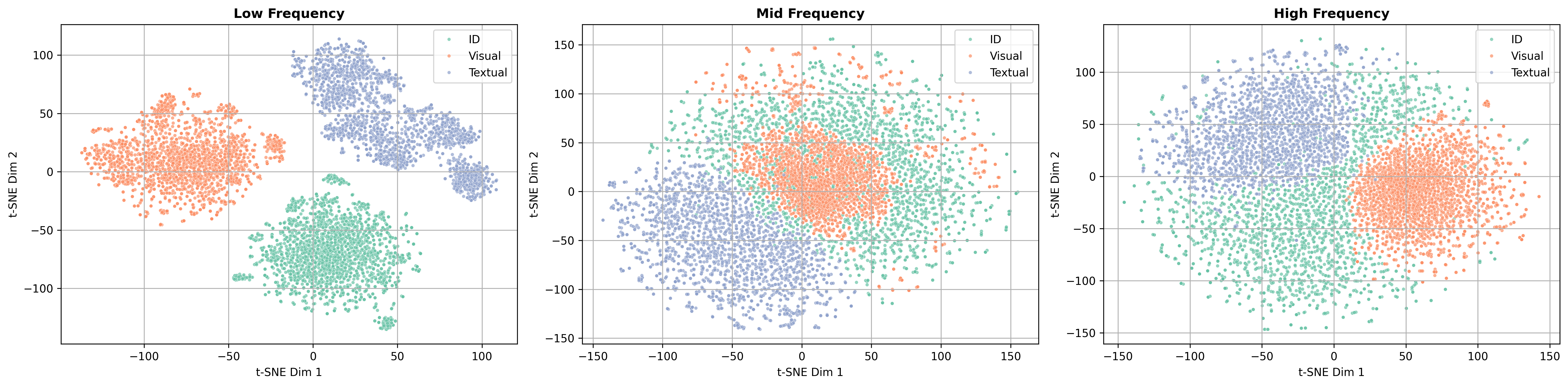}
  \caption{
Spectral visualization of ID, visual, and textual representations across three frequency bands using t-SNE. Low-frequency bands exhibit modality-specific clustering; mid-frequency bands show increasing cross-modal overlap; high-frequency bands reflect task-specific, modality-agnostic semantic encoding, validating our frequency-aware modeling design.
}
  \Description{Modal Frequency Distribution.}
  \label{freq_vis}
\end{figure*}

\begin{figure}[h]
  \centering
  \includegraphics[width=\linewidth]{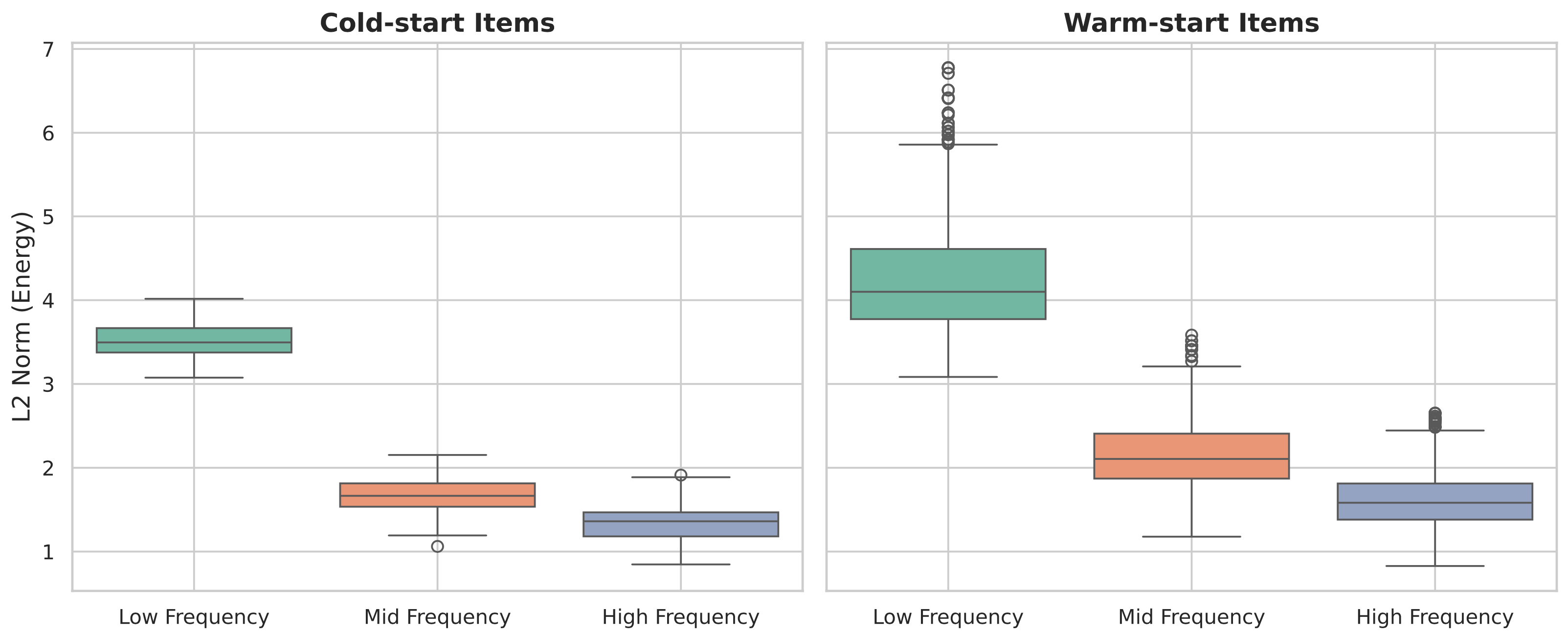}
  \caption{
Comparison of frequency-wise energy distribution between cold-start and warm-start items. 
}
  \Description{Frequency Weight Distribution.}
  \label{feq_band_w}
\end{figure}

\subsection{Spectral Analysis of Cross-Modality Representations (RQ4)}
To examine how the model captures semantics and cross-modal interactions from a frequency-domain perspective, we analyze ID, visual, and textual representations on the Sports dataset, as illustrated in Fig.~\ref{freq_vis}. In the low-frequency band, the three modalities form distinct and well-separated clusters. This suggests that low-frequency components primarily encode modality-specific structures, such as long-term user preferences in the ID modality, visual style in the visual modality, and category-level semantics in the textual modality. The clear separation demonstrates the effectiveness of frequency decomposition in isolating dominant features and supports the use of modality-aware fusion strategies that can selectively prioritize or reduce signals to improve robustness. As frequency increases, modality boundaries gradually blur. In the mid-frequency band, different modalities begin to overlap, reflecting cross-modal interactions and shared semantics, such as common appearance cues between images and text. This pattern supports our model's goal of learning collaborative representations. In the high-frequency band, representations become more entangled and less modality-specific. These components encode sparse but informative details, such as fine-grained item attributes or short-term behavioral patterns, which are critical for personalization. This behavior aligns with our information bottleneck design, which aims to preserve predictive information while reducing irrelevant noise.

\subsection{Frequency Characteristics under Cold-Start and Warm-Start Conditions (RQ5)}

To investigate how different frequency bands contribute under varying data sparsity, we visualize the $L_2$ energy distribution of cold-start and warm-start item representations across frequency components (Fig.~\ref{feq_band_w}). For cold-start items, energy is highly concentrated in the low-frequency band, suggesting the model relies more on stable modality-specific signals (e.g., visual or textual descriptions) to construct coarse-grained preferences in the absence of user history. This supports our design, where low-frequency components retain general and robust semantic cues. In contrast, warm-start items exhibit stronger responses in mid- and high-frequency bands. The mid-frequency band reflects cross-modal semantic alignment, while the high-frequency band captures personalized, fine-grained details derived from rich interaction signals. Notably, the high-frequency energy for cold-start items remains suppressed, indicating our frequency gating and information bottleneck mechanisms effectively filter uninformative components under sparse conditions. These observations confirm that our model adaptively allocates representational focus across frequency bands, enabling robust generalization for cold-start cases and personalize for warm-start items.

\subsection{Hyperparameter Sensitivity (RQ6)}

We investigate the sensitivity of our model to four key hyperparameters, with results shown in Fig.~\ref{s1}. For the information bottleneck weight $\lambda$, performance consistently improves as $\lambda$ increases, peaking at $0.1$ or $1.0$. This indicates that appropriate regularization helps suppress redundant frequency signals and enhances generalization, while too small a $\lambda$ fails to constrain uninformative components effectively. Similarly, the contrastive loss weight $\eta$ achieves the best results at $0.0001$ or $0.01$, suggesting that mild contrastive signals enhance modality alignment and high-frequency discriminability. However, excessive regularization (e.g., $\eta=1.0$) harms performance, likely due to interference with the main optimization objective.

As for the number of frequency decomposition bands $K$, the model performs best at $K=3$ to $5$, validating the importance of moderate decomposition granularity: too coarse fails to capture multi-scale semantics, while too fine leads to sparse and unstable signals. Increasing the latent dimension $d$ generally improves performance, but the optimal value varies across datasets (e.g., $d=256$ for \textit{Baby}, $d=128$ for \textit{Sports}), implying a trade-off between representation capacity and data complexity. 

\begin{figure}[h]
  \centering
  \includegraphics[width=\linewidth]{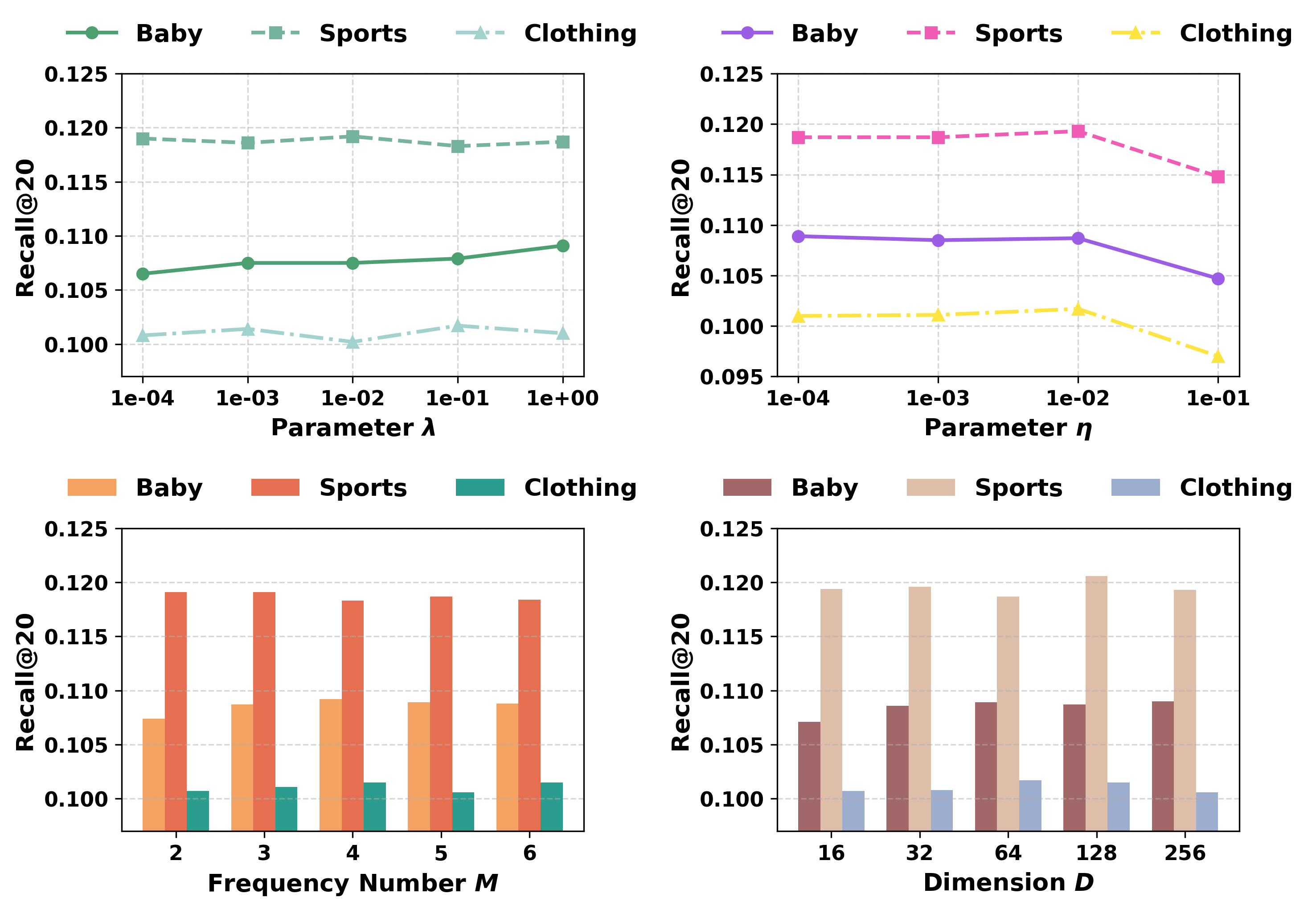}
  \caption{Sensitivity analysis of key hyperparameters. The model shows robust performance under moderate settings.}
  \Description{Different Attention Layer size.}
  \label{s1}
\end{figure}

\section{CONCLUSION}

In this paper, we propose a novel frequency-aware multimodal recommendation framework, FITMM, which effectively integrates adaptive frequency decomposition with an information-theoretic representation learning paradigm. Extensive experiments on three real-world Amazon datasets demonstrate that FITMM consistently outperforms state-of-the-art baselines, particularly in cold-start scenarios. In future work, we plan to further explore personalized frequency modeling, such as user-dependent frequency selection, and investigate the integration of pretrained multimodal foundation models to enhance semantic understanding and scalability.

\balance
\bibliographystyle{ACM-Reference-Format}
\bibliography{sample-base}


\end{document}